\documentclass[runningheads]{llncs}

\usepackage{graphicx}
\usepackage{float}
\usepackage{amsmath}
\usepackage{hyperref}
\usepackage{amssymb}
\usepackage{subfigure}
\usepackage{epstopdf}
\usepackage{url}
\usepackage{array}

\newtheorem{fact}{Fact}
\newcommand{\PreserveBackslash}[1]{\let\temp=\\#1\let\\=\temp}
\newcolumntype{C}[1]{>{\PreserveBackslash\centering}p{#1}}
\newcolumntype{R}[1]{>{\PreserveBackslash\raggedleft}p{#1}}
\newcolumntype{L}[1]{>{\PreserveBackslash\raggedright}p{#1}}

\begin{document}
\title{Evolutionary Equilibrium Analysis for Decision on Block Size in Blockchain Systems
}
\titlerunning{Decision on Block Size by Evolutionary Equilibrium Analysis }
%
\author{Jinmian Chen\and Yukun Cheng\inst{*} \and Zhiqi Xu \and Yan Cao}
\authorrunning{JM. Chen \and YK. Cheng et al.}

\institute{School of Business, Suzhou University of Science and Technology,\\
	 Suzhou, 215009, China\\
\email{jinmian\_chen@post.usts.edu.cn, ykcheng@amss.ac.cn, joisexzq@163.com, 150653659@qq.com}}
\maketitle              
\begin{abstract}
 In a PoW-based blockchain network, mining pools (the solo miner could be regarded as a mining pool containing one miner) compete to successfully mine blocks to pursue rewards. Generally, the rewards include the fixed block subsidies and time-varying transaction fees. The transaction fees are offered by the senders whose transactions are packaged into blocks and is increasing with the block size. However, the larger size of a block brings the longer latency, resulting in a smaller probability of
successfully mining. Therefore, finding the optimal block size to trade off these two factors is a complex and crucial problem for the mining pools. In this paper, we model a repeated mining competition dynamics in blockchain system as an evolutionary game to study the interactions among mining pools. In this game, each pool has two strategies: to follow the default size $\bar{B}$, i.e., the upper bound of a block size, or not follow. Because of the bounded rationality, each mining pool pursues its evolutionary stable block size (ESS) according to the mining pools' computing power and other factors by continuous learning and adjustments during the whole mining process. A study framework is built for the general evolutionary game, based on which we then theoretically explore the existence and stability of the ESSs for a case of two mining pools. Numerical experiments with real Bitcoin data are conducted to show the evolutionary decisions of mining pools and to demonstrate the theoretical findings in this paper.


\end{abstract}

\section{Introduction}
Bitcoin is a decentralized payment system \cite{SN08}, based on a public transaction ledger, which is called the blockchain. Generally, a block is composed of a block header and a block body, which contains a certain amount of transactions. Each transaction is composed of the digital signature of the sender, the transaction data, such as the value of digital tokens, the addresses of the sender and the receiver, as well as the corresponding transaction fee. With the bitcoin system developing, the number of transactions in the whole network increases quickly, while the block size currently is limited to 1 MB. Such a bounded block size results in the congestion of the blockchain network. To alleviate this situation, \emph{Segregated Witness} (SegWit) \cite{EJP15} is brought up and applied to segregate the witness (digital signatures) from the transactions. Then the witness is put into the ``extended block", which has no impact on the original block size. By SegWit, a block is able to contain more transactions, enhancing the transaction processing efficiency. This effect is equivalent to expanding the block size to 2 MB. Thus the block size of a newly mining block may be more than 1 MB.

Proof-of-work (PoW) is the most popular consensus applied in bitcoin blockchain system, which reaches a consensus based on miners' computing power. Under the PoW-based consensus protocol, the process of successfully mining a block includes two steps, i.e., solving the PoW puzzle and propagating the block to be verified. During the propagation, the block is likely to be discarded because of long latency which depends on the size of the block. The larger size of block brings the longer latency, leading to a higher chance that the block suffers \emph{orphaned} \cite{HN16}. So, besides raising income from transaction fees by packaging more transactions in a block, miners need to consider a suitable total size of the block that would not deeply increase the probability to be orphaned. With the incentive of transaction fees and the long latency resulting from large block size, how to select transactions and decide the total block size for maximum payoff is critical for every miner.

Game theory has been widely applied in mining management, such as computational power allocation, fork chain selection, block size setting and pool selection. In terms of block size setting, \cite{PR15} analyzed the quantity setting of block space with the effect of a transaction fee market, in which a block space supply curve and a mempool demand curve were introduced to find the optimal block space for the maximum payoff of miners. By proposing a Bitcoin-unlimited scheme, \cite{ZP17} modeled a non-cooperative game to examine the interaction among the miners, each of whom chooses its own upper bound of the block size while it invalidates and discards the excessive block that is larger than its upper bound. And the game was proved to exist an unique Nash equilibrium where all miners choose the same upper bound. Given the limitation on the number of transactions included in the block, the interaction on choosing transactions between the miners and the users was modeled as a non-cooperative game by \cite{DMS19}. The unique Nash equilibrium of this game can be obtained when satisfying certain conditions, which is related to the number of miners, the hash rate, bitcoin value, the transaction fees, the block subsidy and the cost of the mining. As for dynamic evolutionary behaviors in blockchain, \cite{LWNZW18} respectively modeled the dynamics of block mining selection among pools and pool selection among individual miners as an evolutionary game in a proof-of-work (PoW) blockchain network. It identified the hash rate and the block propagation delay as two major factors resulting in the mining competition outcome. Also, in \cite{PZJFZ21}, evolutionary game was applied to examine the process dynamics of selecting super nodes for transaction verification in the Delegated proof of stake (DPoS) blockchain. The authors found that the strategy of candidates has to do with how much reward they can obtain from the blockchain platform. Inspired by \cite{JW19}, we adopt its novel expected payoff function different with that of \cite{LWNZW18}, to study the block size determination of mining pools in a dynamic process.

In this work, we assume that the mining pools are bounded rational and can adapt their strategy on different block sizes according to the received rewards. Note that the total computing power of each mining pool and the long latency also affect the choices of pools. Accordingly, we model the repeated mining competition as an evolutionary game, where each pool controlling a certain amount of computing power has to decide whether following the default size---2 MB or not, for maximum payoffs. Evolutionary stable strategies are considered to be the solutions of this game. Then we perform theoretical analysis on the existences and corresponding conditions of the ESS for a special case of two mining pools. Finally, simulations are performed to verify the proposed schemes. In addition, we discuss the impact of the hash rate of pools, the unit transaction fee, the unit propagation delay ,as well as the default block size on the strategy decision of mining pools.

The rest of this paper is organized as follows. Section 2 introduces the system model and the reward function of pools. In Section 3, we formulate an evolutionary game model to study the block size selection problem, and particularly analyze a case of two mining pools. Section 4 presents the numerical results and some additional analysis on different factors, and concludes our study.

\section{System Model and and Mining Pool's Expected Reward}

In this paper, we consider the PoW-based blockchain system where there are $n$ mining pools, denoted by $N = \{1,2,\cdots,n\}$, and each contains several miners. All mining pools compete to mine blocks by costing an amount of computing power, and thus to pursue the corresponding rewards. Similar to
 \cite{GKL15} and  \cite{GKWGRC16}, we assume the whole system is in a quasi-static state, meaning no miners join in or leave the system. Under this assumption, each miner keeps its state unchanged, including which mining pool it is in, and how much computing power it has. This leads to the mining pools' constant scale and their total computing power.


By the consensus protocol of PoW, a mining pool, who obtains the reward, must satisfy the following two conditions: it is the first one to solve a proof-of-work puzzle by consuming an amount of computing power and it is also the first one to make its mined block reach the consensus. The expected average block arriving interval is about of $T=600$ seconds, by adjusting the difficulty of the proof-of-work puzzle. The whole mining process in a blockchain system consists of a series of one-shot competitions, in each of which one block is mined. A one-shot competition can be viewed as a non-cooperative game, in which all mining pools are the players and they shall make decisions on the mined blocks' size to maximize their own rewards.

Similar to the model in \cite{JW19}, we compute the expected reward of each mining pool $i$ by regarding its \emph{block finding time} in one-shot competition as a random variable, denoted by $X_i$, which follows the exponential distribution. To be specific, let $B_i$ be the block size decided by pool $i$. Denote $h_i$ ($0\leq h_i\leq 1$) to be the relative computing power of pool $i$, that is the ratio of pool $i$'s computing power to the total computing power in blockchain system. The propagation time of pool $i$'s block is linear with its size $B_i$, that is $q_i=\rho B_i$. It is obvious that pool $i$'s block cannot reach the consensus if the block finding time is less than the propagation time $q_i$.
The mining rate of pool $i$ is denoted by $\lambda_i=\frac{h_i}{T}$, where $T=600$ secs is the average block arriving interval. By the definition of the exponential distribution, the probability density function (PDF) of $X_i$ is
{\small\begin{eqnarray}
 f_{X_i}(t;B_i,\lambda_i)=\left\{ \begin{array}{l}\label{PDFi}
0,t<q_i;\\
\lambda_i e^{-\lambda_i (t-q_i)},t\geq q_i.
\end{array}
\right.
\end{eqnarray}}
and the cumulative distribution function (CDF) of $X_i$ is
{\small\begin{eqnarray*}
F_{X_i}(t;B_i,\lambda_i)=Pr(X_i\leq t)=\left\{ \begin{array}{l}\label{PDFi}
0,~t<q_i;\\
1-e^{-\lambda_i(t-q_i)},~t\geq q_i.
\end{array}
\right.
\end{eqnarray*}}
So, the probability that block finding time of pool $i$ is larger than $t$ is
{\small\begin{eqnarray*}
Pr(X_i> t)=1-F_{X_i}(t;B_i,\lambda_i)=\left\{ \begin{array}{l}\label{PDFi}
1,~t<q_i;\\
e^{-\lambda_i(t-q_i)},~t\geq q_i.
\end{array}
\right.
\end{eqnarray*}}
Define $X$ to be the block finding time among all mining pools. Then $X=\min_{i\in N}\{X_i\}$, i.e. the first time to find a block, and hence
{\small\begin{eqnarray}\label{eq2}
Pr(X> t)&=&\Pi_{i=1,\cdots,n} Pr(X_i> t)=\Pi_{i\in Active(t)} Pr(X_i> t)\nonumber\\
&=&e^{\sum_{i\in Active(t)}[-\lambda_i(t-q_i)]},
\end{eqnarray}}
where $Active(t)=\{i|q_i\leq t\}$ is the pool set, each pool $i$ in which has the propagation time less than time $t$. For convenience, we call each pool $i$ in $Active(t)$ an active pool at time $t$.

From the probability function of (\ref{eq2}), it is not hard to derive the CDF and PDF of random variable $X$ as the follows
{\small\begin{eqnarray}
F_X(t;\bf B, \bf \lambda)&=&1-Pr(X>t)=1-e^{\sum_{i\in Active(t)}[-\lambda_i(t-q_i)]};\label{CDF_X}\\
f_{X}(t;\bf B, \bf \lambda)&=&(\sum_{i\in Active(t)}\lambda_i)e^{\sum_{i\in Active(t)}[-\lambda_i(t-q_i)]},\label{PDF_X}
\end{eqnarray}}
where ${\bf B}=(B_1,\cdots,B_n)$ and $\bf \lambda=(\lambda_1,\cdots,\lambda_n)$.

As stated before, the reward of a mined block comes from two aspects: the fixed subsidies $R$ (e.g., 6.25 BTC for one block currently), and a variable amount of transaction fees.
Particularly, the transaction fees are more dependent on the size of a block, since a block with a larger size contains more transactions.
For the sake of simplicity, we assume the total transaction fee is linearly dependent on the block size, i.e., $\alpha B_i$. This is similar to the suggested pricing standard of transaction fee for users in some token wallets, such as 0.0005 BTC per KB \cite{FeeR}. So the total reward for a block mined by pool $i$ is $R+\alpha B_i$. In addition, the probability that a pool $i$ solves the proof-of-work puzzle at time $t$ is the ratio of its computing power to all other active mining pools at time $t$. Then the reward of mining pool $i$ in expectation at time $t$ is
{\small\begin{eqnarray}\label{payoff_at_t}
E(reward_i|X=t)=\left\{ \begin{array}{l}\label{PDFi}
0,~t<q_i;\\
\frac{h_i}{\sum_{j\in Active(t)}h_j}(R+\alpha B_i),~t\geq q_i.
\end{array}
\right.
\end{eqnarray}}
and then its reward in expectation is expressed as
{\small\begin{eqnarray}
U_i&=&E[E(reward_i|X=t)]\nonumber \\
&=&\int_{-\infty}^{+\infty}E(reward_i|X=t)\cdot f_X(t;\bf B, \bf \lambda)dt \nonumber\\
&=&\lambda_i(R+\alpha B_i)\sum_{l=i}^n\frac{e^{\sum_j \lambda_j(q_j-q_l)}-e^{\sum_j \lambda_j(q_j-q_l+1)}}{\sum_{j\in Active(q_l)}
	\lambda_j},
\end{eqnarray}}
where $q_{n+1}=+\infty$.

\section{Evolutionary Game Model for Decision on Block Size}
In a PoW-based blockchain system, we suppose that there are $n$ independent mining pools, each pool $i$ owning an amount of relative computing power $h_i$. The whole mining process is a series of one-shot competitions, and all mining pools compete to mine a block to win the reward in each one-shot. In this paper, we model the mining competition dynamics as an evolutionary game to study the dynamic interactions among mining pools. In our evolutionary game model, each pool has two kinds of strategies: to follow the default size $\bar{B}$, i.e., the upper bound of a block size, or to choose the block size less than $\bar{B}$. For simlicity, these two strategies is named as "following" strategy and "not following" strategy, respectively.  Because of the bounded rationality, each mining pool pursues its evolutionary stable block size (ESS) through continuous learning and adjustments.

In this section, we first propose the analysis scheme for general case, and then theoretically analyze the existence and stability of the ESS for a case of two mining pools

\subsection{Analysis Scheme}
In our evolutionary game model, there is a crucial problem for each mining pool $i$ that is how to decide the optimal block size to maximize its payoff in expectation. Note that each pool $i$ has two kinds of strategies: one is to fix the block to default size, e.g., $\bar{B}=2 MB$, and the other is to choose a block size $B_i<\bar{B}$. So
in the $k$-th shot, let us define two subsets,
{\small$$N^1(k)=\{i\in N| B_i(k)<\bar{B}\}~\mbox{and}~N^2(k)=N-N^1(k)=\{i\in N|B_i(k)=\bar{B}\},$$}
and call $(N^1(k),N^2(k))$ a \emph{subset profile}. Clearly, subset profile $(N^1(k),N^2(k))$ is determined after all mining pools making decisions on their block sizes. There are $2^n$ subset profiles totally in  each one-shot, and hence we denote the collection of subset profiles in the $k$-th slot by $\mathcal{N}(k)=\{(N^1(k),N^2(k))\}$. 

Suppose that a subset profile $(N^1(k), N^2(k))$ in the $k$-th shot is given.
Each pool $i\in N^1(k)$ selects the "not following" strategy. In addition, it continues to decide the optimal block size $B_i^*<\bar{B}$ by maximizing its expected payoff under a given subset profile $(N^1(k), N^2(k))$.
{\small\begin{eqnarray}
  B_i^*&=&\arg\pi^i_{(N^1(k),N^2(k))}=\arg\max_{B_i< \bar{B}}U_i\nonumber\\
  &=&\arg\max_{B_i< \bar{B}}\left\{ \lambda_i(R+\alpha B_i)\sum_{l=i}^n \frac{e^{\sum_j \lambda_j(q_j-q_l)}-e^{\sum_j \lambda_j(q_j-q_{l+1})}}{\sum_{j\in Active(q_l)} \lambda_j} \right\}.\label{payoff1}
\end{eqnarray}}
Each mining pool $i\in N^2(t)$ sets its block size as $\bar{B}$ and has its payoff
{\small\begin{eqnarray}\label{payoff2}
  \pi^i_{(N^1(k),N^2(k))}=\lambda_i(R+\alpha \bar{B})\sum_{l=i}^n \frac{e^{\sum_j \lambda_j(q_j-q_l)}-e^{\sum_j \lambda_j(q_j-q_{l+1})}}{\sum_{j\in Active(q_l)}\lambda_j}.
\end{eqnarray}}

 During the evolutionary game, the mining pools keep learning to adjust their low-income strategies to a higher-income one dynamically. Until $n$ mining pools reach a stable strategy profile, at which no one would like to change its strategy, an equilibrium state of block size $(B_1^*,B_2^*,\cdots, B_n^*)$ is obtained. Though all the mining competitions are carried out during a series of discrete slots, we can view each block generating slot as a very small interval with respect to the whole mining process, and hence deal with it as a continuous version. It allows us to apply the standard technique to study the evolutionary process for the decisions on block size.

 Let $x_i(k)$, $0\leq x_i(k)\leq 1$, represent the probability of mining pool $i\in N$ to choose the ``not following" strategy at the $k$-th slot. Correspondingly, the probability of pool $i$ to choose the default size is $1-x_i(k)$. If the choice of pool $i$ is not to follow the default size, then its conditional expected payoff is
{\small\begin{eqnarray}
E^{1}_i(k)=\sum_{\substack{(N^1(k),N^2(k))\in \mathcal{N}(k),\\ i\in N^1(k)}}\left(\prod_{l\in N^1(k),l\neq i}x_l(k)\prod_{l\in N^2(k)}(1-x_l(k))\cdot\pi_{(N^1(k),N^2(k))}^i\right).\label{Epayoff1}
\end{eqnarray}}
If mining pool $i$ selects the ``following default size" strategy, then its conditional expected payoff is
{\small\begin{eqnarray}
E^{2}_i(k)=\sum_{\substack{(N^1(k),N^2(k))\in \mathcal{N}(k),\\ i\in N^2(k)}}\left(\prod_{l\in N^1(k)}x_l(k)\prod_{l\in N^2(k),l\neq i}(1-x_l(k))\cdot\pi_{(N^1(k),N^2(k))}^i\right).\label{Epayoff2}
\end{eqnarray}}
Combining (\ref{Epayoff1}) and (\ref{Epayoff2}), the average payoff of mining pool $i$ is
{\small\begin{eqnarray}\label{average_payoff}
\bar{E}_i(k)=x_i(k)E_i^1(k)+(1-x_i(k))E_i^2(k).
\end{eqnarray}}

By \cite{FD98}, the growth rate of a strategy selected by a participant
is just equal to the difference between the payoff of this strategy and its average payoff.
Then the replicator dynamic equations for all mining pools are as follows:
{\small\begin{eqnarray}
  f_i(\mathbf{x})=\dot{x}_i(k)=x_i(k)(E_i^1(k)-\bar{E}_i(k)),~~\forall i\in N.\label{replicatordynamic}
\end{eqnarray}}

According to the replicator dynamics \ref{replicatordynamic}, a mining pool would like to choose a smaller block size, when its conditional payoff $E^1_i(k)$ is larger than the average payoff $\overline{E}_i(k)$. Otherwise, it will set its block size as $\bar{B}$. A state is stable if no mining pool would like to change its strategy over time in the replicator dynamics, and such a stable state is considered to be the evolutionary equilibrium \cite{JGR16}. The strategies in this state are evolutionary stable, called ESS. Specifically speaking, when the payoff of ``not following" strategy is equal to the average payoff for each pool, all mining pools reaches the ESS and no one has incentive to change its current strategy. Therefore, the ESS can be obtained by solving $\dot{x}_i(k)=0$ for all $i\in N$, whose solution is called the \emph{fixed equilibrium point} of replicator dynamics.

\subsection{A Case Study of Two Mining Pools}
Based on the analysis scheme for general case in previous subsection, we continue to study the case of two mining pools $(n=2)$ to exemplify the equilibrium analysis for the decision on block size. We normalize the whole computing power in system, thus mining pool $i$'s relative computing power is  $h_i \in \textbf{h}=\{h_1,h_2\}$ and $\sum_{i=1}^2 h_i=1$. The whole mining process contains a series of one-shot competitions, and pool 1 and pool 2 need to decide their block sizes $B_1$ and $B_2$ to pursue the optimal payoffs in each one-shot competition. Without loss of generality, we concentrate on the case of $0\leq h_1\leq\frac12\leq h_2\leq 1$. The analysis for the case of $h_2\leq h_1\leq$ is symmetric, and thus we omit the discussion. As stated in \cite{JW19}, `` a miner with less mining power prefers a smaller
block size in order to optimize his payoff". Thus the case of $0\leq h_1\leq h_2\leq 1$ leads to $B_1\leq B_2$ in one-shot competition and then the propagation time $q_1=\rho B_1\leq q_2=\rho B_2$.
In a one-shot mining competition, if $B_1\leq B_2\leq\bar{B}$, then
{\small\begin{eqnarray}
U_1&=&\lambda_1(R+\alpha B_1)\sum_{l=1}^2\frac{e^{\sum\lambda_j(q_j-q_l)}-e^{\sum\lambda_j(q_j-q_{l+1})}}{\sum_{j\in Active(q_l)}\lambda_j}\nonumber\\
   &=&(R+\alpha B_1)[1-h_2e^{\lambda_1\rho(B_1-B_2)}];\\
U_2&=&\lambda_2(R+\alpha B_2)\sum_{l=2}^2\frac{e^{\sum\lambda_j(q_j-q_l)}-e^{\sum\lambda_j(q_j-q_{l+1})}}{\sum_{j\in Active(q_l)}\lambda_j}\nonumber\\
   &=&(R+\alpha B_2)h_2e^{\lambda_1\rho(B_1-B_2)}.\label{payoffU1<}
\end{eqnarray}}


Since each mining pool has two kinds of strategies, i.e., to follow the default size $\bar{B}$, and not to follow, in a one-shot competition, there are four strategy profiles: $(B_1, B_2)$, $(B_1, \bar{B})$, $(\bar{B}, B_2)$, and $(\bar{B}, \bar{B})$. Note that $(\bar{B}, B_2)$ either does not exist, or equals to $(\bar{B}, \bar{B})$ under the condition of $B_1\leq B_2\leq\bar{B}$. So we do not discuss this strategy profile any more.

Clearly, each pool would receive different payoffs, under different strategy profiles. For the strategy profile $(B_1,B_2)$, meaning that both pools choose the ``not following" strategy, we define the payoffs of two pools $\pi_{11}^1$ and $\pi_{11}^2$ are their optimal payoffs subject to the conditions of $B_1<\bar{B}$ and $B_2<\bar{B}$. The corresponding optimal block sizes $(B_1^*,B_2^*)$ can be obtained by solving 
{\small $$\frac{\partial U_1(B_1,B_2)}{\partial B_1}=0,~\mbox{and}~\frac{\partial U_2(B_1,B_2)}{\partial B_2}=0,$$} simultaneously. For the strategy profile $(B_1,\bar{B})$, showing that pool 1 would not follow the default size and pool 2's block size is $\bar{B}$, $\pi_{12}^1$ and $\pi_{12}^2$ are denoted to be the payoffs of pool 1 and 2. To be specific, $\pi_{12}^1$ is defined to be the optimal payoff of pool 1 under the condition of $B_1<\bar{B}$ and the corresponding optimal block size $B_1^*$ can be determined by solving $\frac{d U_1(B_1,\bar{B})}{d B_1}=0$. For the strategy profile of $(\bar{B},\bar{B})$, both of two pools set their block sizes as $\bar{B}$, then their payoffs are denoted by $\pi_{22}^1$ and $\pi_{22}^2$.
We illustrate the payoffs of two pools under different strategy profiles in the following payoff matrix (Table \ref{table1}).

\renewcommand\arraystretch{1.5}
{\small\begin{table}[!h]
\begin{center}
\caption{ Payoff matrix of the case of two mining pools.}\label{table1}
\begin{tabular}{ccc}
\hline
\multicolumn{1}{l}{} & \multicolumn{2}{c}{Mining Pool 2}                                                                                                      \\ \cline{2-3}
Mining Pool 1            & $B_2$($x_2$)              & $\bar{B}$(1-$x_2$)
\\ \cline{2-3}
$B_1$($x_1$)         & $(\pi_{11}^1,\pi_{11}^2)$          & $(\pi_{12}^1,\pi_{12}^2)$ \\
$\bar{B}$(1-$x_1$)   & $(\backslash, \backslash)$          & $(\pi_{22}^1,\pi_{22}^2)$ \\
\hline
\end{tabular}
\end{center}
\end{table}}

\begin{lemma}\label{payoff}
	In a one-shot mining competition, if $B_1\leq B_2\leq \bar{B}$, then
	\begin{enumerate}
		\item For strategy profile $(B_1,B_2)$ with $0\leq B_1< B_2< \bar{B}$, the optimal block size of pool 2 is $B_2^*=\frac{1}{\lambda_1\rho}-\frac{R}{\alpha}$, if $0<\frac{1}{\lambda_1\rho}-\frac{R}{\alpha}<\bar{B}$. Let $\widehat{B}_1^*$ be the solution satisfying $\frac{dU_1(B_1,B_2^*)}{d B_1}=0$. If $0\leq \widehat{B}_1^*<B_2^*$, then the optimal block size of pool 1 is $\widehat{B}_1^*$. Then the payoffs are
{\small\begin{eqnarray}
		\pi_{11}^1=(R+\alpha \widehat{B}^{*}_1)[1-h_2e^{\lambda_1\rho(\widehat{B}^{*}_1-B_2^*)}],~
		\pi_{11}^2=[R+\alpha B_2^*]h_2e^{\lambda_1\rho(\widehat{B}^{*}_1-B_2^*)}.\label{payoff111}
		\end{eqnarray}}
If $\widehat{B}_1^*<0$, then the best choice of pool 1 is to set its block size as zero and the payoffs are
 {\small\begin{eqnarray}		\pi_{11}^1=R[1-h_2e^{-\lambda_1\rho B_2^*}],~	\pi_{11}^2=(R+\alpha B_2^*)h_2e^{-\lambda_1\rho B_2^*}.\label{payoff112}
		\end{eqnarray}}
		\item For strategy profile $(B_1,\bar{B})$ with $0<B_1<B_2=\bar{B}$, the block size of pool 2 is $\bar{B}$. Let $\widetilde{B}_1^*$ be the solution satisfying $\frac{dU_1(B_1,\bar{B})}{d B_1}=0$. If $0\leq \widetilde{B}_1^*<\bar{B}$, then the optimal block size of pool 1 is $\widetilde{B}_1^*$. Then the payoffs are
{\small\begin{eqnarray}
		\pi_{12}^1=(R+\alpha \widetilde{B}_1^{*})[1-h_2e^{\lambda_1\rho(\widetilde{B}_1^{*}-\bar{B})}	],~
		\pi_{12}^2=(R+\alpha \bar{B})h_2e^{\lambda_1\rho(\widetilde{B}_1^{*}-\bar{B})}.\label{payoff122}
		\end{eqnarray}}
If $\widetilde{B}_1^*<0$, then the best choice of pool 1 is to set its block size as zero and the payoffs are
 {\small\begin{eqnarray}
		\pi_{12}^1=R[1-h_2e^{-\lambda_1\rho\bar{B}}	],~
		\pi_{12}^2=(R+\alpha \bar{B})h_2e^{-\lambda_1\rho\bar{B}}.\label{payoff122}
		\end{eqnarray}}
		\item For strategy profile $(\bar{B},\bar{B})$, the block sizes of two pools are both equal to $\bar{B}$ and the corresponding payoffs are
		{\small\begin{eqnarray}		\pi_{22}^{1}=h_1(R+\alpha\bar{B})~~\mbox{and}~~\pi_{22}^{2}=h_2(R+\alpha\bar{B}).\label{payoff22}
		\end{eqnarray}}
	\end{enumerate}	
\end{lemma}

\noindent{\bf Remark 1:}
By Lemma \ref{payoff}-(1) and (2), the optimal block size of pool 1 depends on the block size of pool 2 in strategy profiles $(\widehat{B}^{*}_1,B_2^*)$ and $(\widetilde{B}_1^{*},\bar{B})$, and can be obtained from equations $\frac{dU_1(B_1,B_2^*)}{d B_1}=0$ and $\frac{dU_1(B_1,\bar{B})}{d B_1}=0$, if $0\leq\widehat{B}^{*}_1<B_2^*$ and $0\leq \widetilde{B}_1^{*}<\bar{B}$, respectively. For convenience, we denote $g(B_2)$ to be the implicit function, satisfying
{\small\begin{eqnarray*}
\frac{dU_1(B_1,B_2)}{dB_1}
=\alpha-[\alpha+\lambda_1\rho(R+\alpha g(B_2)]h_2 e^{\lambda_1\rho(g(B_2)-B_2)}=0,
\end{eqnarray*}}
Therefore,
$\widehat{B}^{*}_1=g(B_2^*)=g(\frac{1}{\lambda_1\rho}-\frac{R}{\alpha}) ~\mbox{and}~ \widetilde{B}_1^{*}=g(\bar{B}).$

\begin{lemma}\label{g(B_2)}
Let $g(B_2)$ be the implicit function satisfying $\frac{dU_1(g(B_2),B_2)}{dB_1}=0$. Then $g(B_2)$ is monotone increasing with $B_2$ and $g(B_2)< B_2$ for all $B_2\geq 0$.
\end{lemma}

We prove Lemma \ref{g(B_2)} in Appendix A.  Based on the monotonicity of $g(B_2)$, we have $\widehat{B}^{*}_1<\widetilde{B}_1^{*}$, if $\frac{1}{\lambda_1\rho}-\frac{R}{\alpha}<\bar{B}$. Moreover,
the property of $g(B_2)\leq B_2$ ensures $\widehat{B}^{*}_1<\frac{1}{\lambda_1\rho}-\frac{R}{\alpha}$ and $\widetilde{B}_1^{*}<\bar{B}$.



In the following, we would analyze the strategy selections of two mining pools by distinguishing two conditions: (1) $\frac{1}{\lambda_1\rho}-\frac{R}{\alpha}\geq \bar{B}$; (2) $0\leq \frac{1}{\lambda_1\rho}-\frac{R}{\alpha}< \bar{B}$; and (3) $\frac{1}{\lambda_1\rho}-\frac{R}{\alpha}<0$, and then explore the equilibrium solutions in the evolutionary game in the following.
\begin{theorem}\label{NE}
	In a one-shot mining competition, if $0\leq  B_1\leq B_2\leq \bar{B}$ and $\frac{1}{\lambda_1\rho}-\frac{R}{\alpha}
	\geq \bar{B}$, then
	\begin{enumerate}
		\item $(\widetilde{B}_1^{*},\bar{B})$ is a strict Nash equilibrium, if $0\leq \widetilde{B}_1^{*}< \bar{B}$; or
         \item $(0,\bar{B})$ is a strict Nash equilibrium, if $ \widetilde{B}_1^{*}< 0$; or
		\item $(\bar{B},\bar{B})$ is a strict Nash equilibrium, if $ \widetilde{B}_1^{*}\geq \bar{B}$.
	\end{enumerate}
\end{theorem}

Theorem \ref{NE} illustrates that the dominant strategy of pool 2 is to follow the default size $\bar{B}$ when $\frac{1}{\lambda_1\rho}-\frac{R}{\alpha}
	\geq \bar{B}$, while the optimal strategy of pool 1 depends on the value of $\widetilde{B}_1^{*}$. The detailed proof is provided in Appendix B.

Next, we concentrate on the condition of $0\leq \frac{1}{\lambda_1\rho}-\frac{R}{\alpha}
< \bar{B}$, under which pool 2 may set its block size as $B_2^*=\frac{1}{\lambda_1\rho}-\frac{R}{\alpha}$ or $\bar{B}$.

Recall that $x_i(k)$, $i= 1, 2$, is the probability of mining pool $i$ to adopt the ``following" strategy in the $k$-th slot competition, and thus $1-x_i(k)$ is the probability of mining pool $i$ to follow the default size. The expected payoffs of mining pool 1 in the $k$-th slot competition, when it chooses ``not following" or ``follow" strategy, are
{\small\begin{eqnarray*}
  E_1^1(k)=x_2(k)\pi_{11}^1+(1-x_2(k))\pi_{12}^1;~~ E_1^2(k)=(1-x_2(k))\pi_{22}^1.
\end{eqnarray*}}
The average payoff of mining pool 1 is
{\small\begin{eqnarray*}
  \overline{E}_1(k)=x_1(k)E^1_1(k)+(1-x_1(k))E^2_1(k).
\end{eqnarray*}}
Similarly, we can derive the expected payoffs of mining pool 2 as follows,
{\small\begin{eqnarray*}
  E_2^1(k)=x_1(k)\pi_{11}^2;~~ E_2^2(k)=x_1(k)\pi_{12}^2+(1-x_1(k))\pi_{22}^2.
\end{eqnarray*}}
The average payoff of mining pool 2 is
{\small\begin{eqnarray*}
  \overline{E}_2(k)=x_2(k)E^1_2(k)+(1-x_2(k))E^2_2(k).
\end{eqnarray*}}
Based on the analysis scheme (\ref{replicatordynamic}) for the general case, the replicator dynamic system of pool 1 and 2 for the case of two mining pools are:
{\small\begin{eqnarray}
 \left\{ \begin{array}{l}\label{system}
	f_1(\mathbf{x})=\dot{x}_1(k) = x_1(1 - x_1)(E^1_1-\overline{E}_1)=x_1(1 - x_1)\left[(\pi_{11}^1-\pi_{12}^1+\pi_{22}^1)x_2 + (\pi_{12}^1-\pi_{22}^1)\right];\\
	f_2(\mathbf{x})=\dot{x}_2(k) =x_2(1 - x_2)(E^1_2-\overline{E}_2)= x_2(1 - x_2) \left[(\pi_{11}^2-\pi_{12}^2+\pi_{22}^2)x_1-\pi_{22}^2 \right].
	\end{array}
	\right.
\end{eqnarray}}
Note that all the solutions satisfying $f_1(\mathbf{x})=\dot{x}_1(k)=0$ and $f_2(\mathbf{x})=\dot{x}_2(k)=0$ are the fixed equilibrium points of the replicator dynamic system. It is not hard to see that there exist four fixed equilibrium points of this system under the condition of $B_1\leq B_2\leq\bar{B}$: $(0,0)$, $(1,0)$, $(1,1)$ and $(x_1^*,x_2^*)$, where
{\small\begin{eqnarray}\label{partial_fixed}
 x^*_1=\frac{\pi_{22}^2}{\pi_{11}^2-\pi_{12}^2+\pi_{22}^2},~
 x^*_2=\frac{\pi_{22}^1-\pi_{12}^1}{\pi_{11}^1-\pi_{12}^1+\pi_{22}^1}.
\end{eqnarray}}
To fulfill the condition for probability vector $\mathbf{x}$, $x^*_1$ and $x_2^*$ must be in $[0,1]$.

\begin{theorem}\label{theo1}
  For the evolutionary game between two mining pools, if $0\leq B_1\leq B_2\leq\bar{B}$ and $0\leq \frac{1}{\lambda_1\rho}-\frac{R}{\alpha}< \bar{B}$, then
  \begin{itemize}
  	\item  $(1,0)$ is an ESS, when
  (1) $0\leq \widehat{B}_1^*<\widetilde{B}_1^*<\bar{B}$ ~or~ (2) $\widehat{B}_1^*<0<\widetilde{B}_1^*<\bar{B} ~and~ \lambda_1\rho(\frac{R}{\alpha}+\bar{B})
  e^{\lambda_1\rho(\widetilde{B}_1^{*}-\bar{B}-\frac{R}{\alpha})+1}>1$;;
    \item $(1,1)$ is an ESS,  when (1) $ \widehat{B}_1^*<\widetilde{B}_1^*<0$ ~or~ (2) $\widehat{B}_1^*<0<\widetilde{B}_1^*<\bar{B} ~and~ \lambda_1\rho(\frac{R}{\alpha}+\bar{B})
        e^{\lambda_1\rho(\widetilde{B}_1^{*}-\bar{B}-\frac{R}{\alpha})+1}<1$;
    \item $(0,0)$ and $(x_1^*,x_2^*)$ cannot be ESSs.
  \end{itemize}
\end{theorem}

\begin{proof}
To obtain the ESS of the evolutionary game for block size selection, we first compute the Jacobian matrix of the replicator dynamic system (\ref{system}),	\begin{eqnarray*}
	J= \left[ \begin{array}{c}
		\frac{\partial f_1(\mathbf{x})}{\partial x_1}\\
		\frac{\partial f_2(\mathbf{x})}{\partial x_1}
	\end{array} \right.\left.
	\begin{array}{c}
		\frac{\partial f_1(\mathbf{x})}{\partial x_2}\\
		\frac{\partial f_2(\mathbf{x})}{\partial x_2}
	\end{array} \right],
\end{eqnarray*}
where
{\small\begin{eqnarray*}
	\frac{\partial f_1(\mathbf{x})}{\partial x_1} &=& (1-2x_1)\{[\pi_{11}^1-(\pi_{12}^1-\pi_{22}^1)]x_2+(\pi_{12}^1-\pi_{22}^1)\};\\
	\frac{\partial f_1(\mathbf{x})}{\partial x_2} &=& x_1(1-x_1)[\pi_{11}^1-(\pi_{12}^1-\pi_{22}^1)];\\
	\frac{\partial f_2(\mathbf{x})}{\partial x_1} &=& x_2(1-x_2)[(\pi_{11}^2-\pi_{12}^2)+\pi_{22}^2];\\
	\frac{\partial f_2(\mathbf{x})}{\partial x_2} &=& (1-2x_2)\{[(\pi_{11}^2-\pi_{12}^2)+\pi_{22}^2)]x_1-\pi_{22}^2\}.
\end{eqnarray*}}

{\small{\small\begin{table}[H]
\begin{center}
\caption{The determinants and traces of Jacobian matrix $J$ at fixed equilibrium points.}\label{table2}
\begin{tabular}{C{2cm}|C{2cm}|C{2cm}}
\hline
\multicolumn{1}{l}{\textbf{}} & $Det(J)$  & $Tr(J)$                                                                                                                                                                                                              \\ \hline
(0,0)                         & $K\cdot (-M)$    & $K - M$
\\ \hline
(1,0)                         & $(-K)\cdot N$    & $N - K$
\\ \hline
(1,1)                         & $(-L)\cdot (-N)$    & $-L-N$
\\ \hline
$({x_1^*},{x_2^*})$           & $\frac{KLMN}{(L-K)(M+N)}$    & 0                                                                                                                                                                                                                                 \\ \hline
\end{tabular}
\end{center}
\end{table}}}

By the results in \cite{FD98}, if a fixed equilibrium point $(x_1,x_2)$ is an ESS, then the Jacobian matrix of the replicator dynamic system is negative definite at $(x_1,x_2)$, equivalent to determinant $Det (J(x_1,x_2))>0$ and trace $Tr(J(x_1,x_2))<0$.

To simplify the discussion, let us denote
$K=\pi_{12}^1-\pi_{22}^1, L=\pi_{11}^1, M=\pi_{22}^2, N=\pi_{11}^2-\pi_{12}^2. $
Table \ref{table2} shows the determinants and the traces of Jacobian matrix at different fixed equilibrium points. Next we propose the fact, based on which it is easy for us obtain this theorem.

\begin{fact}\label{fact}
	Based on the expressions of $\pi_{11}^i$, $\pi_{12}^i$ and $\pi_{22}^i$, $i=1,2$, in (\ref{payoff111})-(\ref{payoff22}),
	{\small\begin{eqnarray*}
		K &=&\pi_{12}^1-\pi_{22}^1>0,~L=\pi_{11}^1>0,~ M = \pi_{22}^2>0,\\
		N &=& \pi_{11}^2-\pi_{12}^2\\
			&=&   \left\{ \begin{array}{l}
	\left[R+\alpha (\frac{1}{\lambda_1\rho}-\frac{R}{\alpha})\right]h_2e^{\lambda_1\rho(\widehat{B}^{*}_1-\frac{1}{\lambda_1\rho}+\frac{R}{\alpha})}-(R+\alpha \bar{B})h_2 e^{\lambda_1\rho(\widetilde{B}_1^{*}-\bar{B})}<0,\\
 \mbox{if}~ 0\leq \widehat{B}_1^*< \widetilde{B}_1^*<\bar{B};\\
 \left[R+\alpha (\frac{1}{\lambda_1\rho}-\frac{R}{\alpha})\right]h_2e^{\lambda_1\rho(-\frac{1}{\lambda_1\rho}+\frac{R}{\alpha})}-(R+\alpha \bar{B})h_2 e^{\lambda_1\rho(\widetilde{B}_1^{*}-\bar{B})}<0(>0),\\
 \mbox{if}~ \widehat{B}_1^*<0<\widetilde{B}_1^*<\bar{B} ~and~ \lambda_1\rho(\frac{R}{\alpha}+\bar{B})e^{\lambda_1\rho(\widetilde{B}_1^{*}-\bar{B}-\frac{R}{\alpha})+1}>1(<1);\\
	\left[R+\alpha (\frac{1}{\lambda_1\rho}-\frac{R}{\alpha})\right]h_2e^{\lambda_1\rho(-\frac{1}{\lambda_1\rho}+\frac{R}{\alpha})}-(R+\alpha \bar{B})h_2 e^{-\lambda_1\rho\bar{B}}>0,\\
 \mbox{if}~ \widetilde{B}_1^*<\widetilde{B}_1^*\leq 0.
	\end{array}
	\right.
	\end{eqnarray*}}
\end{fact}
Because $Tr (J(x_1^*,x_2^*))=0$ and $Det(J(0,0))=-K\cdot M<0$, $(x_1^*,x_2^*)$ and $(0,0)$ cannot be ESS. Moreover, according to the results of Fact \ref{fact}, we have $Det(J(1,0))>0$ and $Tr(J(1,0))<0$ when $0\leq \widetilde{B}_1^*<\bar{B}$, and then the fixed equilibrium points $(1,0)$ is an ESS. When $\widetilde{B}_1^*<0$, $Det(J(1,1))>0$ and $Tr(J(1,1))<0$. So the fixed equilibrium points $(1,1)$ is an ESS.\qed
\end{proof}

\begin{corollary}
For the evolutionary game between two mining pools, if $0\leq B_1\leq B_2\leq\bar{B}$ and $0\leq \frac{1}{\lambda_1\rho}-\frac{R}{\alpha}< \bar{B}$, then
  \begin{itemize}
  	\item $(\widetilde{B}_1^*,\bar{B})$ is an evolutionary stable strategy profile, when $0\leq\widehat{B}_1^*< \widetilde{B}_1^*<\bar{B}$, or subject to $\lambda_1\rho(\frac{R}{\alpha}+\bar{B})e^{\lambda_1\rho(\widetilde{B}_1^{*}-\bar{B}-\frac{R}{\alpha})+1}>1$ when $0\leq\widehat{B}_1^*<0<\widetilde{B}_1^*<\bar{B}$;
    \item $(0,\frac{1}{\lambda_1\rho}-\frac{R}{\alpha})$ is an evolutionary stable strategy profile, when $\widehat{B}_1^*<\widetilde{B}_1^*<0$, or  subject to $\lambda_1\rho(\frac{R}{\alpha}+\bar{B})e^{\lambda_1\rho(\widetilde{B}_1^{*}-\bar{B}-\frac{R}{\alpha})+1}<1$ when $0\leq\widehat{B}_1^*<0<\widetilde{B}_1^*<\bar{B}$.
  \end{itemize}
\end{corollary}

At last, let us discuss the case that $\frac1{\lambda_1\rho}-\frac{R}{\alpha}<0$.

\begin{theorem}\label{NE_ESS}
 If $0\leq  B_1\leq B_2< \bar{B}$ and $\frac{1}{\lambda_1\rho}-\frac{R}{\alpha}< 0$, then
	 \begin{itemize}
		\item $(0,0)$ is a strict Nash equilibrium when $\widetilde{B}_1^{*}\leq 0$;
		\item $(\bar{B},\bar{B})$ is a strict Nash equilibrium when $ \widetilde{B}_1^{*}\geq \bar{B}$.
        \item when $0<\widetilde{B}_1^{*}<\bar{B}$, then
           \begin{itemize}
             \item $(0,0)$ is an evolutionary stable strategy profile if  $R>
                 \frac{e^{\lambda_1\rho(\widetilde{B}_1^{*}-\bar{B})}}
                 {1-e^{\lambda_1\rho(\widetilde{B}_1^{*}-\bar{B})}}\alpha\bar{B}$;
             \item $(\widetilde{B}_1^{*},\bar{B})$ is an evolutionary stable strategy profile if $R<
                 \frac{e^{\lambda_1\rho(\widetilde{B}_1^{*}-\bar{B})}}
                 {1-e^{\lambda_1\rho(\widetilde{B}_1^{*}-\bar{B})}}\alpha\bar{B}$.
           \end{itemize}
	 \end{itemize}
\end{theorem}

From the condition of $\frac{1}{\lambda_1\rho}-\frac{R}{\alpha}< 0$, we can see that the fixed subsidy $R$ is quite high ($>\frac{\alpha}{\lambda_1\rho}$). Under this condition,
Theorem \ref{NE_ESS} states an possibility that neither of pools would like to choose transactions into their blocks. It means that the mining pools may give up the available transaction fees in hopes of enhancing their chances to win the high fixed subsidy. The detailed proof can be found in Appendix D.

\section{Numerical Experiments and Conclusions}

\subsection{Numerical Experiments}

In this section, we would analyze the case of two mining pools for the decision on block size. Based on the statistic data about Bitcoin blockchain on \cite{BTCBrowser} , we firstly consider the following setting to discuss the influence of default size on the strategy selection of mining pools:
the unit transaction fee $\alpha=1\times10^{-6}$ BTC , the propagation speed $\rho=2.9\times10^{-4} $ s/Byte, the block subsidy $R=6.25$ BTC and the average mining time $T=600$ seconds, and three upper bounds of block size, i.e., $\bar{B}=$1 MB, 2 MB and $3$ MB, as well as the relative computing power of each pool is $h_i\in(0,1)$ and $h_1+h_2=1$.
Let us set $h_1=0.3$ and $h_2=0.7$ and take 0.2, 0.5 and 0.8 respectively as the initial values of $x_i, (i=1,2)$.
Under the setting, we get $\widehat{B}^{*}_1=-8.70\times10^5, \frac{T}{h_1\rho}-\frac{R}{\alpha}=6.47\times10^5$, and $\widetilde{B}_1^{*}=-6.43\times10^5, 6.95\times10^3, 6.69\times10^5$ corresponding to $\bar{B}=1\times10^6, 2\times10^6, 3\times10^6$, respectively. These results satisfy the conditions of $\widehat{B}^{*}_1<\bar{B}, \widetilde{B}_1^{*}<\bar{B}$ and $\frac{T}{h_1\rho}-\frac{R}{\alpha}<\bar{B}$. Fig. \ref{default} illustrates the evolution processes of behaviors of two pools with different default block sizes, verifying the result that $(1,0)$ is an ESS in Theorem \ref{theo1}. It is clear that the speed of convergence to fixed equilibrium point $(1,0)$, i.e., the strategy profile $(\widetilde{B}_1^{*}, \bar{B})$, becomes faster with the increasing of the maximum capacity of a block. 
 Hence larger upper bound $\bar{B}$ brings pool 2 more payoff, stimulating the speed of convergence to the strategy of $\bar{B}$. At the same time, $g'(B_2)>0$ in (\ref{g(B2)}) shows the optimal size of $B_1$ increases with $B_2$, and thus the rate of convergence to $\widehat{B}^{*}_1$ is accelerated by the increasing of $\bar{B}$.
 \begin{figure}[h!]
	\centering
	\subfigure[The evolutionary behavior of pool 1.]{
		\includegraphics[width=5.7cm]{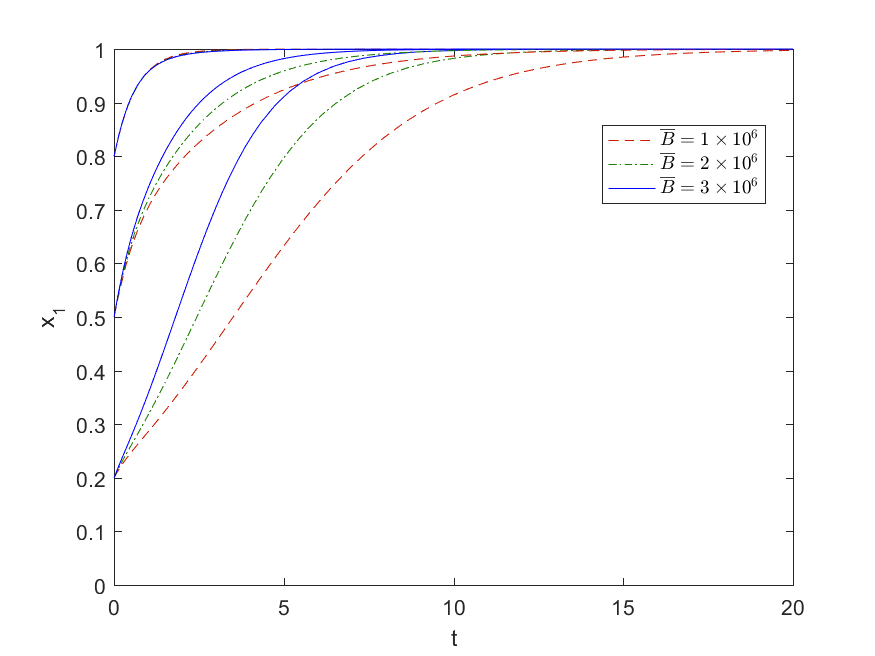}\label{default1}}
	\quad
	\subfigure[The evolutionary behavior of pool 2.]{
		\includegraphics[width=5.7cm]{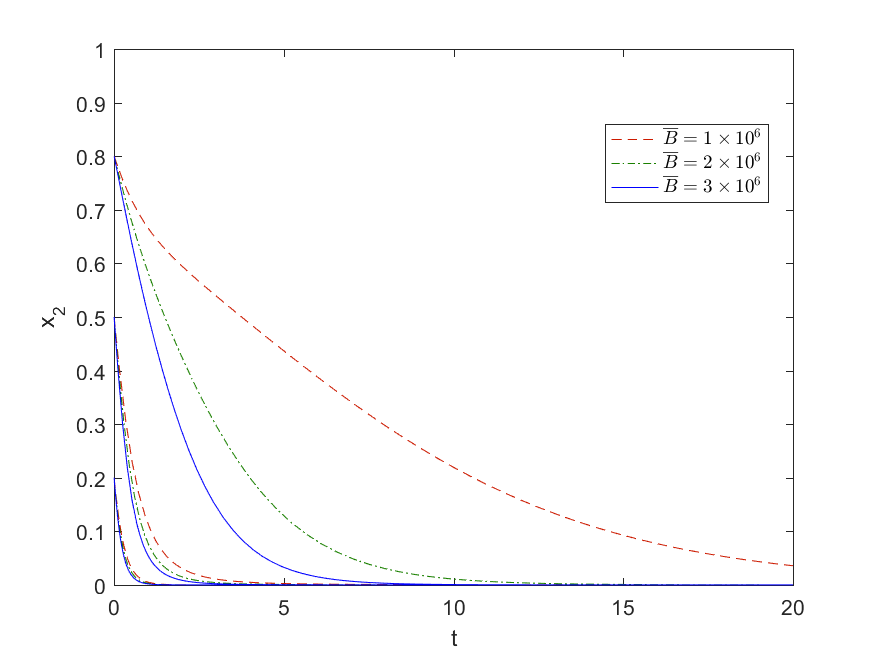}\label{default2}}
	\caption{The impact of the upper bound of block size $\bar{B}$ on pool's behaviors.}\label{default}
\end{figure}



Next, we change the value of $\alpha$ and $\rho$, i.e., $\{\alpha, \rho\}=\{0.8\times10^{-6}, 3\times10^{-4}\}$, and take $h_1=0.24$. Under this setting, $0<B_2^*=\frac{1}{\lambda_1\rho}-\frac{R}{\alpha}<\bar{B}$ and $\widehat{B}^{*}_1<\widetilde{B}_1^{*}<0$. Fig.\ref{ESS11} illustrates the evolutionary behaviors of the two pools are shown, in which the ESS is $(1,1)$, i.e., $(0,\frac{1}{\lambda_1\rho}-\frac{R}{\alpha})$ is the evolutionary stable strategy. From Fig.\ref{ESS11_1}, $x_1$ converges to 1 in a relatively short time, while $x_2$ takes much longer time to converge to 1 shown in Fig.\ref{ESS11_2}. In addition, when $x_2(0)\leq 0.5$, pool 2 has a strong tendency to take $B_2=\bar{B}$ initially, considering the small mining rate of pool 1 and its own large winning probability. However, as time goes by, pool 2 tends to realize the best response is just $B_2^*=\frac{1}{\lambda_1\rho}-\frac{R}{\alpha}$. This explains the transition from $x_2=0\dashrightarrow x_2=1$ when $x_2(0)\leq 0.5$. Hence,if with $\widehat{B}^{*}_1<\widetilde{B}_1^{*}<0$, the best response of pool 1 is $B_1=0$, meaning a block without any transactions just for faster propagation process.
\begin{figure}[H]
	\centering
	\subfigure[The evolutionary behavior of pool 1.]{
		\includegraphics[width=5.7cm]{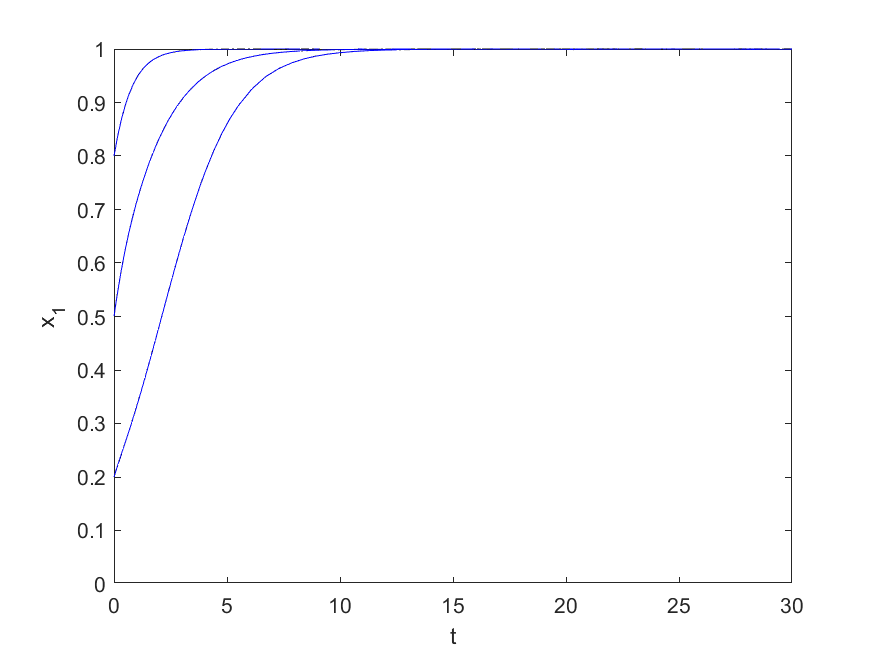}\label{ESS11_1}}
	\quad
	\subfigure[The evolutionary behavior of pool 2.]{
		\includegraphics[width=5.7cm]{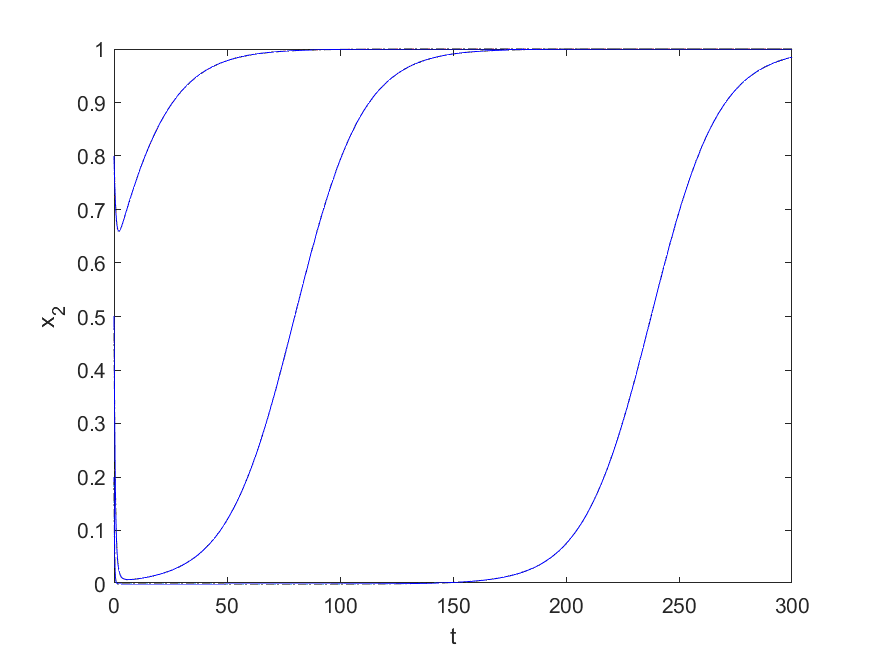}\label{ESS11_2}}
	\caption{The evolution of two pools when (1,1) is ESS.}\label{ESS11}
\end{figure}

\subsection{Conclusion}

In this paper, the issue of selecting appropriate block sizes by mining pools in a blockchain system is discussed. We model this block size determination problem as an evolutionary game, in which each pool may follow the upper bound of block size, i.e., the default size $\bar{B}$, or not. In addition, if a mining pool chooses not to follow $\bar{B}$, then it shall continue to decide its optimal block size under different strategy profile.
In our evolutionary game model, all mining pools are supposed to be bounded rational and each pool switches the low-payoff strategy to a higher one on and on by learning others' better strategies, until the whole network reaches an evolutionary stable state (ESS). The theoretical analysis has been done, particularly for a case of two mining pools, we prove the existence of different ESS under different conditions.
In addition to verify the results in our work, several numerical experiments by using real Bitcoin data are conducted to show the evolutionary decisions of mining pools.

\bibliographystyle{splncs04}
\bibliography{reference}

\begin{thebibliography}{10}
\providecommand{\url}[1]{\texttt{#1}}
\providecommand{\urlprefix}{URL }
\providecommand{\doi}[1]{https://doi.org/#1}

\bibitem{DMS19}
Easley, D., O'Hara, M., Basu, S.: From mining to markets: The evolution of
  bitcoin transaction fees. Journal of Financial Economics  (2019)

\bibitem{FD98}
Friedman, D.: On economic applications of evolutionary game theory. Journal of
  evolutionary economics  \textbf{8}(1),  15--43 (1998)

\bibitem{GKL15}
Garay, J., Kiayias, A., Leonardos, N.: The bitcoin backbone protocol: Analysis
  and applications. In: Advances in Cryptology - EUROCRYPT 2015. pp. 281--310.
  Springer Berlin Heidelberg (2015)

\bibitem{GKWGRC16}
Gervais, A., Karame, G.O., W\"{u}st, K., Glykantzis, V., Ritzdorf, H., Capkun,
  S.: On the security and performance of proof of work blockchains. In:
  Proceedings of the 2016 ACM SIGSAC Conference on Computer and Communications
  Security. pp. 3--16. CCS'16, Association for Computing Machinery (2016)

\bibitem{HN16}
Houy, N.: The bitcoin mining game. Lcloud/fogr Journal  \textbf{1}(13),  53--68
  (2016)

\bibitem{JW19}
Jiang, S., Wu, J.: Bitcoin mining with transaction fees: A game on the block
  size. In: 2019 IEEE International Conference on Blockchain (Blockchain). pp.
  107--115 (2019)

\bibitem{JGR16}
Li, J., Kendall, G., John, R.: Computing nash equilibria and evolutionarily
  stable states of evolutionary games. IEEE Transactions on Evolutionary
  Computation  \textbf{20}(3),  460--469 (2016)

\bibitem{LWNZW18}
Liu, X., Wang, W., Niyato, D., Zhao, N., Wang, P.: Evolutionary game for mining
  pool selection in blockchain networks. IEEE Wireless Communications Letters
  \textbf{7}(5),  760--763 (2018)

\bibitem{EJP15}
Lombrozo, E., Lau, J., Wuille, P.: Segregated witness (consensus layer) (2015),
  \url{https://github.com/bitcoin/bips/wiki/Comments:BIP-0141}

\bibitem{SN08}
Nakamoto, S.: Bitcoin: A peer-to-peer electronic cash system (2008),
  \url{http://bitcoin.org}

\bibitem{FeeR}
NervosFans: Bitcoin transaction fee rules,
  \url{https://zhuanlan.zhihu.com/p/38479785}

\bibitem{PZJFZ21}
Pan, D., Zhao, J.L., Fan, S., Zhang, Z.: Dividend or no dividend in delegated
  blockchain governance: A game theoretic analysis. Journal of Systems Science
  and Systems Engineering pp. 1861--9576 (2021)

\bibitem{PR15}
Rizun, P.R.: A transaction fee market exists without a block size limit (2015),
  block Size Limit Debate Working Paper

\bibitem{BTCBrowser}
TOKENVIEW: Bitcoin browser, \url{https://btc.tokenview.com/}

\bibitem{ZP17}
Zhang, R., Preneel, B.: On the necessity of a prescribed block validity
  consensus: Analyzing bitcoin unlimited mining protocol. In: Proceedings of
  the 13th International Conference on Emerging Networking EXperiments and
  Technologies. pp. 108--119. Association for Computing Machinery, New York,
  NY, USA (2017)

\end{thebibliography}

\newpage
\section*{Appendix}

\subsection*{A. Proof of Lemma \ref{g(B_2)}}
\begin{proof}
Because $g(B_2)$ satisfies
{\small\begin{eqnarray*}
\frac{dU_1(B_1,B_2)}{dB_1}=\alpha-[\alpha+\lambda_1\rho(R+\alpha g(B_2))]h_2 e^{\lambda_1\rho(g(B_2)-B_2)}=0,
\end{eqnarray*}}
we have
{\small\begin{eqnarray}\label{g(B2)}
	1>g'(B_2)=\frac{\alpha +\lambda_1\rho(R+\alpha g(B_2))}{2\alpha +\lambda_1\rho(R+\alpha g(B_2))}>0.
\end{eqnarray}}
So $g(B_2)$ is monotonically increasing with respect to $B_2$. Moreover, when $B_2=0$, $U_1(B_1,0)=(R+\alpha B_1)[1-h_2e^{\lambda_1\rho B_1}]$, and for all $B_1\geq 0$,
{\small\begin{eqnarray*}
\frac{d^2U_1(B_1,0)}{dB_1^2}<0,~\mbox{and}~\frac{dU_1(B_1,0)}{dB_1}|_{B_1=0}<0.
\end{eqnarray*}}
It means that function $U_1(B_1,0)$ is strictly concave, and the maximal point $g(0)$ is less than 0. Let us define an auxiliary function $G(B_2)=g(B_2)-B_2$. Since $g'(B_2)<1$, $G(B_2)$ decreases on interval $[0,\infty)$. Thus $G(B_2)\leq G(0)<0$ for all $B_2\geq 0$, indicating $g(B_2)<B_2$.   \qed
\end{proof}

\subsection*{B. Proof of Theorem \ref{NE}}
\begin{proof}
 Since $\frac{\partial U_2}{\partial B_2}|_{B_2=\frac{1}{\lambda_1\rho}-\frac{R}{\alpha}}=0$ and $\frac{\partial^2 U_2}{\partial B_2^2}<0$ for any block size of $B_1$, $U_2$ is strictly concave and achieves its maximality at the boundary point of $\bar{B}$ subject to the condition of $\frac{1}{\lambda_1\rho}-\frac{R}{\alpha}
 \geq \bar{B}$. Thus if $\frac{1}{\lambda_1\rho}-\frac{R}{\alpha}
 \geq \bar{B}$, then $\bar{B}$ is the dominant strategy of pool 2. In addition, given pool 2's strategy of $\bar{B}$, the payoff function $U_1$ only depends on $B_1$, denoted by $U_{1}(B_1,\bar{B})$. Since $\widetilde{B}_1^{*}$ is the solution of $\frac{dU_1(B_1,\bar{B})}{dB_1}=0$, $U_1$ reaches its maximality at $\widetilde{B}_1^{*}$ if $0\leq \widetilde{B}_1^{*}< \bar{B}$. Moreover,
 {\small\begin{eqnarray*}
 	\frac{d^2 U_1(B_1,\bar{B})}{d B_1^2}=-h_2\lambda_1\rho[2\alpha+\lambda_1\rho(R+\alpha B_1)] e^{\lambda_1\rho(B_1-\bar{B})}<0,
 \end{eqnarray*}}
 showing $U_1$ is strictly concave with respect to $B_1$.
 Therefore,
  {\small\begin{eqnarray}\label{limit}
 \pi_{12}^1=U_1(\widetilde{B}_1^{*},\bar{B})> \lim_{B_1\rightarrow \bar{B}^{-}}U_1(B_1,\bar{B})=h_1(R+\alpha\bar{B})=\pi_{22}^1,
 \end{eqnarray}}
meaning the best response of pool 1 is $\widetilde{B}_1^{*}$. Hence, $(\widetilde{B}^{*}_1,\bar{B})$ is a strict Nash equilibrium.
If $\widetilde{B}_1^{*}\geq \bar{B}$, then
{\small\begin{eqnarray}
 \pi_{22}^1=h_1(R+\alpha \bar{B})=\lim_{B_1\rightarrow \bar{B}^{-}}U_1(B_1,\bar{B})>max_{B_1\in [0,\bar{B})}U_1(B_1,\bar{B}),
 \end{eqnarray}}
by the strict concavity of $U_1$.
It means the best response of pool 1 is $\bar{B}$. So $(\bar{B},\bar{B})$ is a strict Nash equilibrium, if $\widetilde{B}_1^{*}\geq \bar{B}$.
When $\widetilde{B}_1^{*}< 0$, the strict concavity of $U_1$ ensures that $U_1(B_1,\bar{B})$ decreases when $B_1\in [0, \bar{B}]$, and thus
$U_1(0,\bar{B})=\max_{B_1\in [0,\bar{B}]}U_1(B_1,\bar{B})$. Therefore, the best response of pool 1 is $0$ and $(0,\bar{B})$ is a strict Nash equilibrium if $\widetilde{B}_1^{*}< 0.$
  \qed	
\end{proof}

\subsection*{C. Proof of Fact 1.}

\begin{proof}
It is not hard to see $L>0$ and $M>0$. By the concavity of $U_1(B_1,\bar{B})$ shown in the proof of Theorem \ref{NE}, $U_1(B_1,\bar{B})$ is monotonically decreasing when $B_1\in [\widetilde{B}_1^*,\bar{B}]$. Hence, if $0\leq \widetilde{B}_1^*<\bar{B}$, then $\pi_{12}^1=U_1(\widetilde{B}_1^*,\bar{B})>U_1(\bar{B},\bar{B})=\pi_{22}^1$; and if $\widetilde{B}_1^*<0$, then $\pi_{12}^1=U_1(0,\bar{B})>U_1(\bar{B},\bar{B})=\pi_{22}^1$. So $K>0$.

If $0\leq \widehat{B}_1^*< \widetilde{B}_1^*<\bar{B}$, let us define
{\small\begin{eqnarray*}
 F(B_2)=(R+\alpha B_2)h_2 e^{\lambda_1\rho(B_1^*-B_2)}=(R+\alpha B_2)h_2 e^{\lambda_1\rho(g(B_2)-B_2)}.
\end{eqnarray*}}
 Then
{\small\begin{eqnarray}
	\frac{dF(B_2)}{dB_2} &=& h_2\left[\alpha +\lambda_1\rho(R+\alpha B_2)(g'(B_2)-1)\right]e^{\lambda_1\rho(g(B_2)-B_2)}\nonumber\\
	&=&h_2\alpha e^{\lambda_1\rho(g(B_2)-B_2)}\left[\frac{2-\lambda_1\rho(B_2-g(B_2))}{2\alpha +\lambda_1\rho(R+\alpha g(B_2))}\right]>0,\label{F(B2)}
\end{eqnarray}}
where the inequality comes from $0<\lambda_1\rho(B_2-g(B_2))=h_1\frac{\rho(B_2-g(B_2))}{T}<\frac{\rho B_2}{T}<1$, because the propagation time $\rho B_2$ is no more than the interval $T=600sec$ in each one-shot competition. $F(B_2)$ is monotonically increasing with respect to $B_2$. Therefore, $\pi_{11}^2=F(B_2^*)<F(\bar{B})=\pi_{12}^2$, and $N<0$.

If $\widehat{B}_1^*<\widetilde{B}_1^*<0$,
let us define  $H(B_2)=(R+\alpha B_2)h_2 e^{-\lambda_1\rho B_2}$. Then
{\small\begin{eqnarray}
	\frac{dH(B_2)}{dB_2} &=& h_2e^{-\lambda_1\rho B_2}\left[\alpha -(R+\alpha B_2)\lambda_1\rho\right].
\end{eqnarray}}
Clearly, $H(B_2)$ decreases when $B_2\in [B^*_2,\bar{B}]$, and thus $\pi_{11}^2=F(B_2^*)>F(\bar{B})=\pi_{12}^2$. So we have $N=\pi_{11}^2-\pi_{12}^2>0$.

If $\widehat{B}_1^*<0<\widetilde{B}_1^*<\bar{B}$, then
{\small \begin{eqnarray}
  N&=&\left[R+\alpha(\frac1{\lambda_1\rho}-\frac{R}{\alpha})\right]h_2
  e^{-\lambda_1\rho(\frac1{\lambda_1\rho}-\frac{R}{\alpha})}-
  [R+\alpha\bar{B}]h_2e^{\lambda_1\rho(\widetilde{B}_1^*-\bar{B})}\nonumber\\
  &=&h_2e^{-\lambda_1\rho(\frac1{\lambda_1\rho}-\frac{R}{\alpha})}
  \left[\frac{\alpha}{\lambda_1\rho}-(R+\alpha\bar{B})
  e^{\lambda_1\rho(\widetilde{B}_1^*-\bar{B}-\frac{R}{\alpha})+1}\right].\nonumber
\end{eqnarray}}
Clearly,
\begin{eqnarray*}
  N = \pi_{11}^2-\pi_{12}^2
			= \left\{ \begin{array}{l}
<0,~\mbox{if}~ \lambda_1\rho(\frac{R}{\alpha}+\bar{B})e^{\lambda_1\rho(\widetilde{B}_1^{*}-\bar{B}-\frac{R}{\alpha})+1}>1;\\
>0,~\mbox{if}~ \lambda_1\rho(\frac{R}{\alpha}+\bar{B})e^{\lambda_1\rho(\widetilde{B}_1^{*}-\bar{B}-\frac{R}{\alpha})+1}<1.
	\end{array}
	\right.
\end{eqnarray*}
\qed
\end{proof}

\subsection*{D. Proof of Theorem \ref{NE_ESS}.}

\begin{proof}
 Since $\frac{\partial U_2}{\partial B_2}|_{B_2=\frac{1}{\lambda_1\rho}-\frac{R}{\alpha}}=0$ and $\frac{\partial^2 U_2}{\partial B_2^2}<0$ for $0\leq B_1\leq \bar{B}$, $U_2$ is strictly concave and achieves its maximality at $B_2=0$ subject to $\frac{1}{\lambda_1\rho}-\frac{R}{\alpha}\leq 0$ and $0\leq B_2\leq \bar{B}$, thus $U_2(B_1,0)>U_2(B_1,\bar{B})$ for any $0\leq B_1\leq \bar{B}$.

Let us consider the utility function of pool 1 $U_(B_1,0)$ when $B_2=0$. Then
 {\small\begin{eqnarray*}
 	\frac{d U_1(B_1,0)}{d B_1}&=&\alpha-h_2e^{\lambda_1\rho B_1}[\alpha+\lambda_1\rho(R+\alpha B_1)]
 \leq \alpha[1-h_2e^{\lambda_1\rho B_1}(2+\lambda_1\rho B_1)]\\
 &<& \alpha(1-2h_2)<0;\\
 \end{eqnarray*}}
where the first inequality is true from the condition of $\frac{1}{\lambda_1\rho}-\frac{R}{\alpha}\leq 0$, the second one is correct due to $0<\lambda_1\rho B_1<1$, and the last one is from the assumption that
$1/2<h_2<1$. Therefore, $U_1(B_1,0)$ achieves its maximality at point $(0,0)$.
This result also verifies that the conclusion that `` a miner with less mining power prefers a smaller block size in order to optimize his payoff" in \cite{JW19}.

Recall that $\widetilde{B}_1^*$ is obtained from $\frac{d U_1(B_1,\bar{B})}{d B_1}=0$. Because $\frac{d^2 U_1(B_1,\bar{B})}{d B_1^2}<0$, $U_1(B_1,\bar{B})$ is strictly concave on $B_1$ when $B_2=\bar{B}$. Thus, $B_1=0$ is the maximal point if $\widetilde{B}_1^*\leq 0$; $B_1=\widetilde{B}_1^*$ is the maximal point if $0<\widetilde{B}_1^*<\bar{B}$; and $B_1=\bar{B}$ achieves the maximum of $U_1(B_1,\bar{B})$ if $\widetilde{B}_1^* \geq \bar{B}$.

To be specific, when $\widetilde{B}_1^*\leq 0$, the block size game is similar to a Prisoner's Dilemma (Table \ref{PDgame1}). It is obvious that the best response of pool 1 to $B_2=\bar{B}$ is $B_1=0$, due to $U_1(0,\bar{B})>U_1(\bar{B},\bar{B})$. However, for any $B_1\geq0$, the optimal strategy of pool 2 is $B_2=0$. Hence, the strategy profile $(0,0)$ is a strict NE if $\widetilde{B}_1^*\leq 0$.
 {\small\begin{table}[H]
\begin{center}
\caption{The payoff matrix under the condition of $\widetilde{B}_1^*\leq 0$.}\label{PDgame1}
\begin{tabular}{C{2cm}|C{4cm}|C{4cm}}
\hline
  & $B_2=0$    & $B_2=\bar{B}$                                                                                                                                                                                                           \\ \hline
$B_1=0$        & $U_1(0,0),U_2(0,0)$    & $U_1(0,\bar{B}),U_2(0,\bar{B})$
\\ \hline
$B_1=\bar{B}$  & $\backslash, \backslash$  & $U_1(\bar{B},\bar{B}),U_2(\bar{B},\bar{B})$
\\ \hline
\end{tabular}
\end{center}
\end{table}}

When $\widetilde{B}_1^*\geq \bar{B}$, $U_1(B_1,\bar{B})$ reaches its maximum at point $B_1=\bar{B}$. Then the corresponding payoff matrix is shown in Table \ref{PDgame2}. Since $U_i(0,0)<U_i(\bar{B},\bar{B}), i=1,2$, the strategy profile $(\bar{B},\bar{B})$ is a strict NE if $\widetilde{B}_1^*\geq \bar{B}$.

{\small\begin{table}[H]
\begin{center}
\caption{The payoff matrix under the condition of $\widetilde{B}_1^*\geq \bar{B}$.}\label{PDgame2}
\begin{tabular}{C{2cm}|C{4cm}|C{4cm}}
\hline
  & $B_2=0$    & $B_2=\bar{B}$                                                                                                                                                                                                           \\ \hline
$B_1=0$        & $Rh_1,Rh_2$    & $\backslash, \backslash$
\\ \hline
$B_1=\bar{B}$  & $\backslash, \backslash$  & $(R+\alpha \bar{B})h_1,(R+\alpha \bar{B})h_2$
\\ \hline
\end{tabular}
\end{center}
\end{table}}

At last, we analyze the case of $0<\widetilde{B}_1^*<\bar{B}$ in the evolutionary game. For strategy profile $(0,0)$, the payoffs are
$$\pi_{11}^1=Rh_1,\pi_{11}^2=Rh_2 ;$$
For strategy profile $(\widetilde{B}_1^*,\bar{B})$, the payoffs are
$$\pi_{12}^1=(R+\alpha \widetilde{B}_1^*)[1-h_2e^{\lambda_1 \rho(\widetilde{B}_1^*-\bar{B})}],\pi_{12}^2=(R+\alpha \bar{B})h_2e^{\lambda_1 \rho(\widetilde{B}_1^*-\bar{B})} ;$$
For strategy profile $(\bar{B},\bar{B})$, the payoffs are
$$\pi_{22}^1=(R+\alpha \bar{B})h_1,\pi_{22}^2=(R+\alpha \bar{B})h_2. $$

Then $L=\pi_{11}^1>0, M=\pi_{22}^2>0$, $K=\pi_{12}^1-\pi_{22}^1$ and $N=\pi_{11}^2-\pi_{12}^2$.
Furthermore, we define function $H(x)=(R+\alpha x)[1-h_2 e^{\lambda_1 \rho(x-\bar{B})}], x\in(0,\bar{B})$, thus $K=H(\widetilde{B}_1^*)-H(\bar{B})$, then

{\small\begin{eqnarray*}
\frac{d H(x)}{dx}&=&\alpha-h_2 e^{\lambda_1 \rho(x-\bar{B})}[\alpha+\lambda_1 \rho(R+\alpha x)];\\
\frac{d^2 H(x)}{dx^2}&=&-h_2\lambda_1 \rho e^{\lambda_1 \rho(x-\bar{B})}[2\alpha+\lambda_1 \rho(R+\alpha x)]<0.
\end{eqnarray*}}
Note that $x=\widetilde{B}_1^*$ satisfies $\alpha-h_2 e^{\lambda_1 \rho(x-\bar{B})}[\alpha+\lambda_1 \rho(R+\alpha x)]=0$, i.e., $H'(\widetilde{B}_1^*)=0$. With $\frac{d^2 H(x)}{dx^2}<0$, there exists $H(\widetilde{B}_1^*)>H(\bar{B})$, i.e., $K>0$. And
{\small\begin{eqnarray*}
 N&=&\pi_{11}^2-\pi_{12}^2\left\{ \begin{array}{l}
    >0, if \frac{R}{R+\alpha\bar{B}} > e^{\lambda_1\rho(\widetilde{B}_1^{*}-\bar{B})};\\
    <0, if \frac{R}{R+\alpha\bar{B}} < e^{\lambda_1\rho(\widetilde{B}_1^{*}-\bar{B})}.
 \end{array} \right.
 \end{eqnarray*}}
So, the fixed point $(1,1)$, corresponding to the evolutionary stable strategy profile $(0,0)$, is an ESS if $\frac{R}{R+\alpha\bar{B}} > e^{\lambda_1\rho(\widetilde{B}_1^{*}-\bar{B})}$; the fixed point $(1,0)$, corresponding to the evolutionary stable strategy profile $(\widetilde{B}_1^*,\bar{B})$, is an ESS if $\frac{R}{R+\alpha\bar{B}} < e^{\lambda_1\rho(\widetilde{B}_1^{*}-\bar{B})}$.
\qed
\end{proof}

\end{document}